\documentclass[twoside,11pt]{article}

\usepackage{jmlr2e}

\usepackage{epsf}
\usepackage{fancyhdr}
\usepackage{graphics}
\usepackage{graphicx}
\usepackage{psfrag}

\usepackage[linesnumbered,ruled]{algorithm2e}
\DontPrintSemicolon

\usepackage{color}

\usepackage{amsthm}
\usepackage{amsfonts}
\usepackage{amsmath}
\usepackage{amssymb,bbm}

\usepackage{nicefrac}

\usepackage{chngpage}

 \usepackage{tabularx}%

\usepackage{enumitem}
\usepackage{booktabs}
\usepackage{caption,subcaption}

\usepackage{mathtools}

\newcommand{\braces}[1]{\left\{ #1 \right \}}

\newtheorem{proposition}{Proposition}
\newtheorem{lemma}{Lemma}

\ShortHeadings{Generalized Inverse Gaussian Distribution Generator}{Zhang and Reiter}
\firstpageno{1}

\begin{document}

\title{A Generator for Generalized Inverse Gaussian Distributions}

\author{\name Xiaozhu Zhang \email xiaozhu.zhang@duke.edu \\
       \addr Department of Statistical Science\\
       Duke University\\
       Durham, NC 27708, USA
       \AND
       \name Jerome P. Reiter \email jreiter@duke.edu \\
       \addr Department of Statistical Science\\
       Duke University\\
       Durham, NC 27708, USA}

\editor{TBD }

\maketitle

\begin{abstract}

We propose a new generator for the generalized inverse Gaussian (GIG) distribution by decomposing the density of GIG into two components. The first component is a truncated inverse Gamma density, in order to sample from which we improve the traditional inverse CDF method. The second component is the product of an exponential pdf and an inverse Gamma CDF. In order to sample from this quasi-density, we develop a rejection sampling procedure that adaptively adjusts the piecewise proposal density according to the user-specified rejection rate or the desired number of cutoff points. The resulting complete algorithm enjoys controllable rejection rate and moderate setup time. It preserves efficiency for both parameter varying case and large sample case.

\end{abstract}


\section{Introduction}

The generalized inverse Gaussian (GIG) distribution was first formulated by~\cite{halphen1941nouveau}, yet it can be traced back to 1915 when Schr\"odinger and Smoluchowski independently derived the density of a certain type of Brownian motion~\citep[see][]{seshadri2012inverse}. This distribution was rediscovered and then popularized in the seventies by~\cite{barndorff1977infinite}, \cite{barndorff1977exponentially}, \cite{barndorff1978first} and \cite{barndorff1978hyperbolic} through several applications that modeled the size of diamonds, incomes and turbulence. Recently the GIG distribution has become popular in quantitative risk management of financial engineering~\citep{mcneil2015quantitative}. In addition to the broad applications, the theoretical properties of GIG were further investigated by~\cite{jorgensen1982statistical} and Seshadri~\cite{seshadri2012inverse}.

The density of generalized inverse Gaussian (GIG) distribution is given by
$$
g(y|\lambda, \psi, \chi) = \frac{(\chi/ \psi)^{\lambda/2} }{2 K_\lambda(\sqrt{\psi\chi})} y^{\lambda-1} \exp\left( -\frac{1}{2} \left( \frac{\chi}{y} + \psi y \right) \right) \cdot \mathbf{1}_{\{y>0\}} ,
$$
where $K_\lambda$ denotes the modified Bessel function of the third kind with index $\lambda$~\citep[see][]{abramowitz1948handbook}. The domain of variation for the parameters is
\begin{align}\label{for:range}
\lambda \in \mathbb{R}, \quad
(\psi, \chi) \in 
\begin{cases}
\braces{ (\psi, \chi): \psi >0, \chi \geq 0 }, & \text{if $\lambda >0$},  \\
\braces{ (\psi, \chi): \psi >0, \chi > 0 }, & \text{if $\lambda = 0$}, \\
\braces{ (\psi, \chi): \psi \geq 0, \chi > 0 }, & \text{if $\lambda < 0$}.
\end{cases}
\end{align}
Particularly, the gamma distribution is a special case of this family with $(\chi = 0, \lambda > 0)$, and the inverse gamma distribution is a special case of this family with $(\psi = 0, \lambda < 0)$. Thus if we exclude these two cases we get an alternative parametrization by setting $\alpha = \sqrt{\psi/\chi}$ and $\beta = \sqrt{\psi\chi}$: Let $X = \alpha Y$, we have
$$
g'(x|\lambda, \beta) \propto x^{\lambda-1} \exp\left( -\frac{\beta}{2} (x + \frac{1}{x}) \right) \cdot \mathbf{1}_{\{x>0\}}.
$$
where $\beta > 0$. In addition, we know that if $X \sim \mathcal{GIG}(\lambda, \psi, \chi)$, then $1/X \sim \mathcal{GIG}(-\lambda, \chi, \psi)$.

Many approaches have been developed to generate variates from this distribution since it is proposed. The most direct approach is based on the traditional inverse CDF method, such as the Gauss-Lobatto integration and Newton interpolation proposed by~\cite{leydold2011generating}. The inverse CDF is approximated by polynomials with user-specified tolerated error, although its setup procedure is expensive and thus cannot be efficiently applied to varying parameter case such as inside a Gibbs sampler. 

Another series of GIG generators falls into the category of rejection sampling. The early work includes the multi-part exponential proposal density developed by \cite{atkinson1982simulation}, but the algorithm requires to numerically solve a bivariate optimization problem and the rejection rate is less than 50\% when $\lambda \leq 0.5$ and $\beta \leq 0.1$. Perhaps the most popular sampler is the Ratio-of-Uniforms method proposed by~\cite{dagpunar1989easily}. It works well when $\lambda \geq 1$ or $\beta \geq 0.5$ but does not have a uniformly bounded rejection constant. This algorithm is available as \texttt{ghyp::rgig()} on CRAN. In addition to that, \cite{dagpunar2007simulation} established a gamma proposal density by maximizing the overall acceptance rate. The algorithm itself is easy to implement but the rejection rate is greater than 50\% only when $\beta < \lambda$.

Fortunately, based on the successful Ratio-of-Uniforms method, \cite{hormann2014generating} established an improved version specifically for the difficult case $\lambda < 1$ and $\beta < 0.5$. This new procedure manages to balance the acceptance rate with a fast setup and induces a uniformly bounded rejection constant. This algorithm is also available on CRAN as \texttt{GIGrvg::rgig()}, which is deemed to be the most efficient function that generates GIG variates in current R packages according to our numerical experiments.

In this paper we propose a new generation method for the GIG distribution that adaptively controls the rejection rate given a user-specified upper bound. This new procedure chooses the complexity of proposal density automatically according to the desired rejection rate, and thus provides a tool for users to manually and freely balance between the acceptance rate and setup efficiency. The paper is organized as follows: In Section 2, we derive the GIG distribution as a marginal distribution of two components (1) a truncated inverse Gamma density, and (2) a quasi-density as the product of an exponential PDF and an inverse gamma CDF. In order to sample from the first component, in Section 3, we improve the inverse CDF method that samples from a truncated Gamma (or truncated inverse Gamma) distribution. In Section 4, we develop a rejection sampling procedure, which based on a used-specified rejection rate, adaptively adjusts the proposal density to sample from the second component. In Section 5, we formulate a complete sampling algorithm for GIG by combining the procedures developed in previous sections, and then we conduct time complexity analysis and make comparisons with other existing algorithms.

\section{GIG as a Marginal Distribution}
In this section we derive the GIG density as a marginal density of two components when $\lambda < 0$. For any $X\sim \mathcal{GIG}(\lambda > 0, \psi, \chi)$, we may regard it as $X = 1/Z$ where $Z \sim \mathcal{GIG}(-\lambda < 0, \chi, \psi)$. Therefore, this marginal property of GIG is applicable to the whole domain of $\lambda$ except for the point 0. 

\begin{proposition}\label{pro:marginal}
If 
\begin{enumerate}[label = (\arabic*)]
    \item The random variable $Y$ follows a distribution with density 
    \begin{equation}\label{for:target-den}
    p_Y(y) \propto \exp(-\frac{\beta}{2} y) \cdot \int_{0}^y x^{\lambda-1}  \exp(-\frac{\beta}{2x}) dx ;
    \end{equation}
    
    \item  The random variable $X|Y \sim \mathcal{IG}(-\lambda, \beta/2) \cdot \mathbf{1}_{\{ x\in(0,y) \}} $ where $\lambda <0, \beta >0$.
\end{enumerate}
Then $X \sim \mathcal{GIG}(\lambda, \beta, \beta)$.
\end{proposition}

\begin{proof}
The pdf of $X$ is:
\begin{align*}
\begin{split}
    p_X(x) 
    &= \int_{0}^\infty p(x|y) p(y) dy \\
    &= C \int_{0}^\infty \frac{1}{F(y)} \cdot x^{\lambda-1} \exp(-\frac{\beta}{2x}) \cdot \mathbf{1}_{\{0<x<y\}} \cdot \exp(-\frac{\beta}{2} y) \cdot F(y) dy \\
    &= C x^{\lambda-1} \exp(-\frac{\beta}{2x})  \int_{x}^\infty \exp(-\frac{\beta}{2} y) dy \\
    &= C' x^{\lambda-1} \exp\left( -\frac{\beta}{2} (x + \frac{1}{x}) \right),
\end{split}
\end{align*}
where $C'$ is a normalizing constant not associated with $x$. 
\end{proof}

Proposition \ref{pro:marginal} shows that we are able to sample from any GIG with $\lambda <0$ if we know (1) how to sample from the truncated distribution $\mathcal{IG}(-\lambda, \beta/2) \cdot \mathbf{1}_{x\in(0,y)}$, which is discussed in Section 3; and (2) how to sample from density \ref{for:target-den}, which is addressed in Section 4.

\section{Sampling From a Truncated Inverse Gamma Distribution}

In this section, we focus on the sampling algorithm of truncated inverse Gamma distributions with the form $\mathcal{IG}(a,b) \cdot \mathbf{1}_{\{ 0 < x< t \}}$.
Note that the inverse Gamma distribution can be regarded as the reciprocal Gamma distribution, so it is equivalent to sample $ X \sim \mathcal{G}(a,b) \cdot \mathbf{1}_{\{ x > 1/t \}}$ since $1/X \sim \mathcal{IG}(a,b) \cdot \mathbf{1}_{\{ 0 < x< t \}}$.

Perhaps the most popular and common approach to sample from a uni-variate truncated distribution is the inverse-CDF method. Suppose that $F(\cdot)$ is the CDF of $\mathcal{G}(a,b)$, then for any uniformly distributed random variable $U\sim \mathcal{U}(F(1/t), 1)$, we have $F^{-1}(U) \sim \mathcal{G}(a,b) \cdot \mathbf{1}_{\{ x > 1/t \}}$. However, when $F(1/t)$ is so close to 1 such that the \texttt{pgamma(1/t)} function in R gives exactly 1, for example, the inverse CDF method is impractical. 

We propose an improved inverse CDF method (see Algorithm \ref{algo:inverse-cdf}) to overcome the obstacles described above.

\begin{algorithm}[!htb] 
    \caption{Improved Inverse CDF method \texttt{truncated.gamma()}} 
    \label{algo:inverse-cdf}
    \KwIn{Shape parameter $a$, rate parameter $b$, truncated position $t$;}
    Find $p$ such that $F(t) = 1 - \exp(p)$;
    \\
    Sample $Y\sim \mathsf{Ex}(1)$; \\
    Let $V \leftarrow Y - p$; \\
    Find $X$ such that $V = -\log(1 - F(X))$; \\
    \KwOut{A random variate $X\sim \mathcal{G}(a,b) \cdot \mathbf{1}_{\{ x > t \}} $.  }
\end{algorithm}

The following lemma provides validity for the Algorithm \ref{algo:inverse-cdf}.

\begin{lemma} \label{lem:trunc}
\ 
\begin{enumerate}[label = (\arabic*)]

    \item Suppose that $Y \sim \mathsf{Ex}(1)$, then $Y-p \sim \mathsf{Ex}(1) \cdot \mathbf{1}_{\{x > -p\}}$.
    
    \item Suppose that $-\log(1-U) \sim \mathsf{Ex}(1) \cdot \mathbf{1}_{\{x > -p\}}$, then $U\sim \mathsf{Unif}(1 - e^p, 1)$.
\end{enumerate}
\end{lemma}

Specifically, the step 1 in Algorithm \ref{algo:inverse-cdf} can be achieved by the command \texttt{pgamma(q = 1/t, shape = a, rate = b, lower.tail = FALSE, log.p = TRUE)} in R. The step 3 in Algorithm \ref{algo:inverse-cdf} generates a variate $V$ that shares the same distribution of $-\log(1 - U)$ if $U\sim \mathsf{Unif}(F(t), 1)$, and so $F^{-1}(1 - \exp(-V))$ generates the desired truncated inverse gamma variate $X$. However, when $V$ is extremely large, the function \texttt{qgamma(1 - exp(-V), shape = a, rate = b)} might returns exactly 1. Hence, we have to generate $X$ through the log scale by the command \texttt{qgamma(p = -V, shape = a, rate = b, lower.tail = FALSE, log.p = TRUE)}.

\section{Rejection Sampling}

In this section, we focus on the sampling algorithm of $Y$ with the quasi-density
$$
f(y) = \frac{\beta}{2} \exp(-\frac{\beta}{2} y) \cdot 
\frac{\int_{0}^y x^{\lambda-1}  \exp(-\frac{\beta}{2x}) dx }
{\int_{0}^\infty x^{\lambda-1}  \exp(-\frac{\beta}{2x}) dx }
= h(y) \cdot F(y),
$$
where $h(y)$ is the density of $\mathsf{Ex}(\beta / 2)$, whereas $F(y)$ is the CDF of $\mathcal{IG}(-\lambda, \beta / 2)$. 

\begin{figure}[ht]
    \centering
    \includegraphics[width=0.6\textwidth]{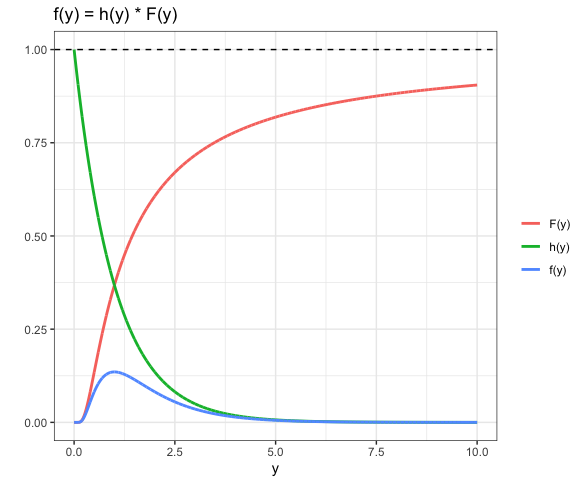}
    \caption{The plot of $f(y) = h(y) F(y)$ when $\lambda = -1$, $\beta = 2$. }
    \label{fig:f(y)}
\end{figure}

\subsection{A na\"ive version}
In order to sample $Y$ with the quasi-density $f(y)$, we may choose $h(y)$ as a proposal density and treat $F(y)$ as the acceptance rate. However, as shown in Figure \ref{fig:f(y)}, when $y$ is close to 0, which should occur very frequently according to the green curve, the acceptance rate will be close to 0, resulting in an unacceptably high rejection rate. The table \ref{tab:acc-rate-1} shows the acceptance rates for different $(\lambda, \beta)$ if $h(y)$ is chosen as the proposal density. Generally speaking, the acceptance rate would go up as $\beta$ decreases or $|\lambda|$ increases; however, the acceptance rate would be around only 0.5\% when $\beta = 0.1$ and $\lambda$ = -0.001.

\begin{table}[h]
\centering
\caption{The mean acceptance rates (standard deviation) with $N = 30$ replicates for different $(\beta, \lambda)$ using the naive version. }
\begin{tabular}{lcccc}
\hline
\multicolumn{1}{c}{} & $\lambda = -0.001$ & $\lambda = -0.01$ & $\lambda = -0.1$ & $\lambda = -1$ \\ \hline
$\beta$ = 0.0001     & 0.018 (0.000)      & 0.171 (0.002)     & 0.845 (0.003)    & 1.000 (0.000)  \\
$\beta$ = 0.001      & 0.014 (0.000)      & 0.131(0.001)      & 0.754 (0.004)    & 1.000 (0.000)  \\
$\beta$ = 0.01       & 0.009 (0.000)      & 0.090 (0.001)     & 0.610 (0.003)    & 1.000 (0.000)  \\
$\beta$ = 0.1        & 0.005 (0.000)      & 0.047 (0.000)     & 0.385 (0.002)    & 0.986 (0.001)  \\ \hline
\end{tabular}
\label{tab:acc-rate-1}
\end{table}

\subsection{A tighter envelope}
We have to design a much nicer proposal density. Suppose that $\mathbf k \in\mathbb{R}^K$, $K\geq 1$ is a vector of cutpoints with $0 < \mathbf k_i < \mathbf k_j < \infty$ for any $i < j$, then we define a step function $F^*(y)$ as an envelope of $F(y)$:
$$
F^*(y) = \sum_{i=0}^{K} F(\mathbf k_{i+1}) \cdot \mathbf{1}_{\braces{ \mathbf k_i \leq y < \mathbf k_{i+1}  }}
$$
with $\mathbf k_0 = 0$ and $\mathbf k_{K+1} = +\infty$. Since $F^*(y) \geq F(y)$ for any $y$, the function $f^*(y) = h(y)F^*(y)$ is also an envelope of $f(y)=h(y)F(y)$ and thus it can work as a proposal density. In addition, the function $f^*(y)$ should be much smaller than the na\"ive proposal density $h(y)$ as $F^*(y) < 1$ for any $y < \mathbf k_{K}$. For example, in Figure \ref{sfig:tigher}, the blue curve is the much tighter proposal density $f^*(y)$. Compared with the red curve $h(y)$, particularly when $y$ is close to 0, the rejection region is clearly much smaller.

\begin{figure}[ht]
     \centering
     \begin{subfigure}[b]{0.49\textwidth}
         \centering
         \includegraphics[width=\linewidth]{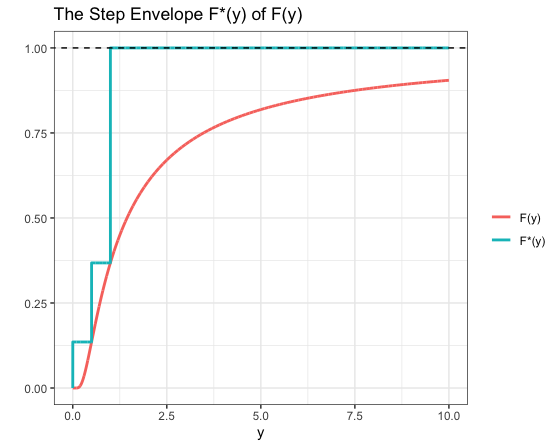}
         \caption{Approximating $F(y)$ by step function $F^*(y)$. }
         \label{sfig:step}
     \end{subfigure}
     \hfill
     \begin{subfigure}[b]{0.49\textwidth}
         \centering
         \includegraphics[width=\linewidth]{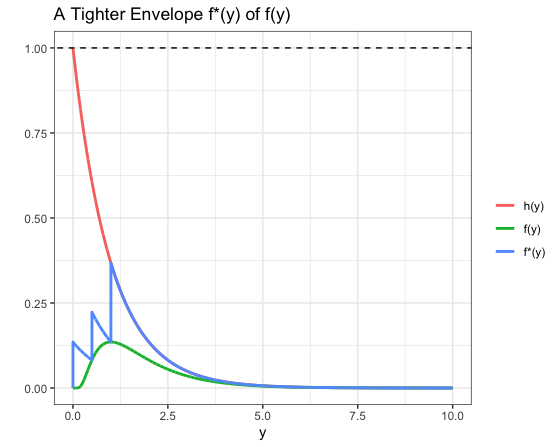}
         \caption{A tighter proposal density $f^*(y)$ than $h(y)$. }
         \label{sfig:tigher}
     \end{subfigure}
     \caption{A tight proposal density $f^*(y) = h(y) \cdot F^*(y)$ when $\lambda = -1$, $\beta = 2$ and the cutpoints are $(0.5, 1, 1.5)^\top$. }
     \label{fig:envelope}
\end{figure}

The proposal density $f^*(y)$ is also a piecewise function:
$$
f^*(y) = \sum_{i=0}^{K} F(\mathbf k_{i+1}) h(y) \cdot \mathbf{1}_{\braces{ \mathbf k_i \leq y < \mathbf k_{i+1}  }}.
$$
For each piece, $F(\mathbf k_{i+1}) h(y) \cdot \mathbf{1}_{\{ \mathbf k_i \leq y < \mathbf k_{i+1} \}}$ is proportional to the density of the truncated exponential distribution $\mathsf{Ex}(\beta / 2) \cdot \mathbf{1}_{\{ k_i \leq y < k_{i+1} \}}$. In addition, since 
$$
\frac{f^*(y)}{ \int_\mathbb{R} f^*(y) dy} =
\frac{\sum_{i=0}^{K} F(\mathbf k_{i+1}) h(y) \cdot \mathbf{1}_{\braces{ \mathbf k_i \leq y < \mathbf k_{i+1}  }}}
{\sum_{i=0}^{K} F(\mathbf k_{i+1})  \int_{\mathbf k_i}^{\mathbf k_{i+1}} h(y) dy  } =
\sum_{i=0}^{K} z_i \cdot \frac{h(y)}{\int_{\mathbf k_i}^{\mathbf k_{i+1}} h(y) dy} \cdot \mathbf{1}_{\braces{ \mathbf k_i \leq y < \mathbf k_{i+1}  }}
$$
where 
$$
z_i = \frac{ F(\mathbf k_{i+1}) \int_{\mathbf k_i}^{\mathbf k_{i+1}} h(y) dy } 
{  \sum_{i=0}^{K} F(\mathbf k_{i+1})  \int_{\mathbf k_i}^{\mathbf k_{i+1}} h(y) dy  },
$$
we may sample any variate $Y$ with the quasi-density $f^*(y)$ by: 
\begin{enumerate}[label = \underline{Step \arabic*}:]
    \item Sample $Z \sim \mathsf{Multinomial}(z_0, \cdots, z_{K})$;
    \item Sample $Y|Z=z \sim \mathsf{Ex}(\beta/2) \cdot \mathbf{1}_{\{ \mathbf k_z \leq y < \mathbf k_{z+1} \}}$.
\end{enumerate}
Furthermore, the rejection rate for any $y\in\mathbb{R}_+$ is given by
$$
\frac{f(y)}{f^*(y)} = \frac{h(y)F(y)}
{ \sum_{i=0}^{K} F(\mathbf k_{i+1}) h(y) \cdot \mathbf{1}_{\braces{ \mathbf k_i \leq y < \mathbf k_{i+1}  }} } = 
\sum_{i=1}^K \frac{F(y)}{F(\mathbf k_{i+1})} \cdot \mathbf{1}_{\braces{ \mathbf k_i \leq y < \mathbf k_{i+1}  }}.
$$

\begin{algorithm}[!htb] 
    \caption{Rejection Sampling \texttt{rejection.sampling()}} 
    \label{algo:rej-samp}
    \KwIn{Number of variates $N$, parameters $\lambda$, $\beta$ and cutoff points $\mathbf k = (\mathbf k_1, \cdots,\mathbf  k_K)$ with $\mathbf k_0 = 0$ and $\mathbf k_{K+1} = +\infty$;}
    Initialization: $T = 0$; \\
    Find $F(\mathbf k_{i+1})$ and $\int_{\mathbf k_1}^{\mathbf k_{i+1}} h(y)dy$ for $i\in\braces{0,\cdots,K}$;
    \\
    Find 
    $$
    z_i =  \frac{F(\mathbf k_{i+1}) \int_{\mathbf k_i}^{\mathbf k_{i+1}} h(y) dy }  {\sum_{i=0}^{K} F(\mathbf k_{i+1})  \int_{\mathbf k_i}^{\mathbf k_{i+1}} h(y) dy}
    $$ for $i\in\braces{0,\cdots,K}$ \\
    \While{$T < N$}{
        Sample $Z \sim \mathsf{Multinomial}(z_0, \cdots, z_{K})$ and let $(l, u) \leftarrow (\mathbf k_Z, \mathbf k_{Z+1})$ ; \\
        Sample $Y' \sim \mathsf{Ex}(1)$ and let $Y \leftarrow Y' \mod (u-l) + l$; \\
        Find the rejection rate $R(Y)$ where
        $$
        R(y) = \sum_{i=1}^K \frac{F(y)}{F(\mathbf k_{i+1})} \cdot \mathbf{1}_{\braces{ \mathbf k_i \leq y < \mathbf k_{i+1}  }} ;
        $$ \\
        Sample $U \sim \mathsf{Unif}(0,1)$; \\
        \If{$U \leq R(Y)$}{
        $Y_T \leftarrow Y$; \\
        $T \leftarrow T + 1$; \\
        }
    }
    \KwOut{A collection of random variates $\{Y_1, \dots, Y_N \}$ with the quasi-density $f(y)$.  }
\end{algorithm}

The Algorithm \ref{algo:rej-samp} demonstrates the construction of this piecewise envelope and the corresponding rejection sampling procedure given specific cutoff points $\mathbf{k}$. In addition, the Lemma \ref{lem:rej-samp} shows that the random variate $Y$ generated by step 6 in Algorithm \ref{algo:rej-samp} follows the distribution $Y|Z=z \sim \mathsf{Ex}(\beta/2) \cdot \mathbf{1}_{\{ \mathbf k_z \leq y < \mathbf k_{z+1} \}}$.

\begin{lemma}
\label{lem:rej-samp}
\ 
\begin{enumerate}[label = (\arabic*)]
    \item Suppose that $Y \sim \mathsf{Ex}(\lambda) \cdot \mathbf{1}_{0< y < b-a}$, then $Y + a \sim \mathsf{Ex}(\lambda) \cdot \mathbf{1}_{\{ a< y < b \}}$;
    
    \item Suppose that $Y \sim \mathsf{Ex}(\lambda)$, then $Y \mod a \sim \mathsf{Ex}(\lambda) \cdot \mathbf{1}_{\{ 0 < y< a \}}$.
\end{enumerate}
\end{lemma}

In order to exam the correctness of the rejection sampling procedure in Algorithm \ref{algo:rej-samp} when we use $f^*(y)$ as the proposal density, we sample from the distribution $\mathcal{GIG}(-0.1, 1, 1)$ with 100000 draws and compare their summary statistics with the actual density. The Figure \ref{fig:check-density} shows the histogram of these 10000 draws, which perfectly aligns with the blue actual density curve. In addition, in Table \ref{tab:compare-stat}, we contrast the actual quantiles and the actual mean which are computed by numerical integrals, with the statistics computed based on the 100000 draws. The simulated ones are similar to the actual ones, which justifies our algorithm \ref{algo:rej-samp} as a valid independent GIG variates generator.

\begin{figure}[ht]
    \centering
    \includegraphics[width=0.6\textwidth]{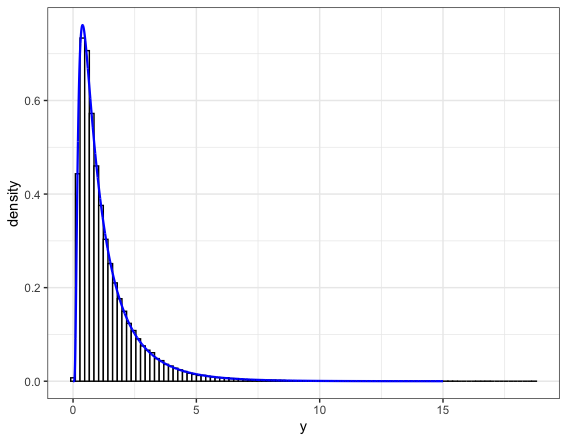}
    \caption{The histogram of varieties from $\mathcal{GIG}(-0.1, 1, 1)$ generated using Algorithm \ref{algo:rej-samp}. The blue curve is the actual density. }
    \label{fig:check-density}
\end{figure}

\begin{table}[ht]
\centering
\caption{The comparisons between actual statistics and simulated statistics computed from the 100000 draws of $\mathcal{GIG}(-0.1, 1, 1)$. }
\begin{tabular}{lllllll}
\hline
Quantiles            & \multicolumn{1}{c}{10\%} & \multicolumn{1}{c}{25\%} & \multicolumn{1}{c}{Median} & \multicolumn{1}{c}{Mean} & \multicolumn{1}{c}{75\%} & \multicolumn{1}{c}{90\%} \\ \hline
Actual statistics    & 0.3045                & 0.5048                & 0.9235                  & 1.3325                 & 1.7020                  & 2.8672                 \\
Simulated statistics & 0.3049                & 0.5081                & 0.9204                  & 1.3309                 & 1.7069                & 2.8682                \\ \hline
\end{tabular}
\label{tab:compare-stat}
\end{table}

\subsection{Cutoff points given rejection rate}

Although the tighter envelope $f^*(y)$ reduces the rejection rate, its particular construction as well as the resulting rejection rate depends on the choice of cutoff points. In this section, we will develop an algorithm that adaptively chooses the cutoff points given a user-specified rejection rate $\epsilon_0$.

\begin{figure}[ht]
     \centering
     \begin{subfigure}[b]{0.49\textwidth}
         \centering
         \includegraphics[width=\linewidth]{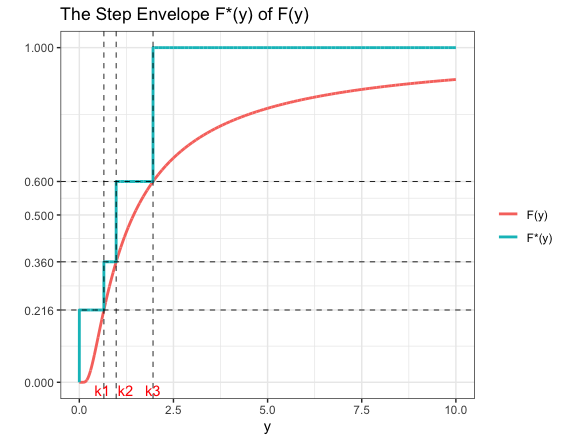}
         \caption{Choosing cutoff points by $F(\mathbf k_3) = 1 - \epsilon_0/2$, $F(\mathbf k_2) = (1 - \epsilon_0/2)^2$ and $F(\mathbf k_3) = (1 - \epsilon_0/2)^3$. }
         \label{sfig:fp-step}
     \end{subfigure}
     \hfill
     \begin{subfigure}[b]{0.49\textwidth}
         \centering
         \includegraphics[width=\linewidth]{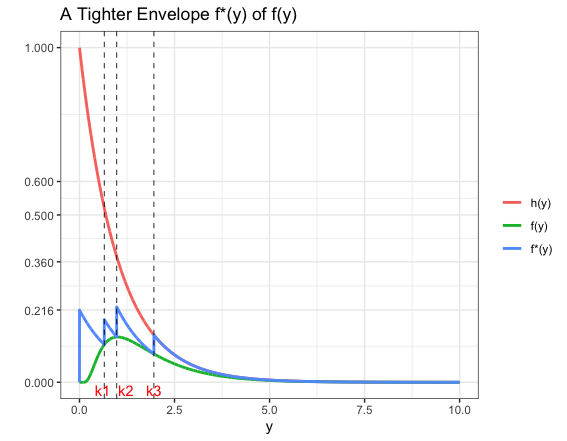}
         \caption{The resulting proposal density whose first piece is chosen with probability less than $\epsilon_0/2$. }
         \label{sfig:fp-tigher}
     \end{subfigure}
     \caption{The process of choosing cutoff points given rejection rate when $\lambda = -1$, $\beta = 2$ and $\epsilon_0 = 0.8$. }
     \label{fig:find-point}
\end{figure}

We illustrate our strategy through an example in the Figure \ref{fig:find-point} where $\lambda = -1$, $\beta = 2$ and $\epsilon_0 = 0.8$:
\begin{enumerate}[label = \underline{Step \arabic*}:]
    \item Find the point $\mathbf k_3$ such that $F(\mathbf k_3) = 1 - \epsilon_0 / 2$, and set $\mathbf k_3$ as a cutoff point. Hence for any $y\in[\mathbf k_3, \infty)$, the rejection rate is controlled under $\epsilon_0/2$;
    
    \item Find the point $\mathbf k_2$ such that $F(\mathbf k_2) = (1-\epsilon_0/2)^2$, and set $\mathbf k_2$ as a cutoff point. Hence for any $y \in[\mathbf k_2, \mathbf k_3)$, the rejection rate is also controlled under $\epsilon_0/2$;
    
    \item Find the point $\mathbf k_1$ such that $F(\mathbf k_1) = (1-\epsilon_0/2)^3$, and set $\mathbf k_1$ as a cutoff point. Hence for any $y \in[\mathbf k_1, \mathbf k_2)$, the rejection rate is also controlled under $\epsilon_0/2$;
    
    \item Now we have $z_0 \leq \epsilon_0/2 $ (where $z_0$ is in the Step 3 of Algorithm \ref{algo:rej-samp}, the proportion of area of the first piece for $f^*(y)$), so the overall rejection rate is less than
    $$
    \frac{\epsilon_0}{2} + (1 - \frac{\epsilon_0}{2}) \cdot \frac{\epsilon_0}{2} \leq \epsilon_0.
    $$
    
\end{enumerate}

The strategy describes a procedure where we start from the density $h(y)$, find a cutoff point $\mathbf k^*$, pull down the $[0, \mathbf k^*)$ part of the density by the constant $(1 - \epsilon_0/2)$, and repeat this on a finer resolution. We formulate this procedure into the following algorithm:
\begin{algorithm}[!htb] 
    \caption{Find cutoff points given rejection rate\\ \texttt{find.cutoff.under.rej.rate()} } 
    \label{algo:find-points}
    \KwIn{Parameters $\lambda$, $\beta$ and rejection rate $\epsilon_0$;}
    Initialization: $\alpha \leftarrow 1 - \epsilon_0 /2 $, $A_l \leftarrow 1$, $A_r \leftarrow 0$, $k_1 \leftarrow \infty$, ${\rm CUTPOINTS} \leftarrow [\ ]$;
    \\
    \While{ $A_l > (A_l + A_r) \cdot \epsilon_0 / 2$ }{
        $\mathbf k_0 \leftarrow F^{-1}(\alpha)$; \\
        ${\rm CUTPOINTS} \leftarrow [\mathbf k_0, {\rm CUTPOINTS}]$; \\
        $\alpha \leftarrow \alpha \cdot (1 - \epsilon_0 / 2)$; \\
        Let
        $$
        P_{A_l} \leftarrow \frac{\int_0^{\mathbf k_0} h(y) dy}{\int_0^{\mathbf k_1} h(y) dy};
        $$ \\
        $A_r \leftarrow A_r + (1-P_{A_l}) \cdot A_l$;
        \\
        $A_l \leftarrow A_l \cdot P_{A_l} \cdot (1 - \epsilon_0 / 2) $; \\
        $k_1 \leftarrow k_0$;
    }
        \KwOut{A vector of cutoff points CUTPOINTS.  }
\end{algorithm}

We have to show that there must exist a time such that $A_l \leq (A_l + A_r) \cdot \epsilon_0 / 2$; in other words, we have to ensure that the Algorithm \ref{algo:find-points} will terminate itself and the resulting vector has finitely many elements. Now denote the potential new cutoff point $\mathbf k_0$ at the $n$-th iteration by the sequence $\braces{x_n}$ where $x_n = F^{-1}((1 - \epsilon_0 / 2)^n)$, then the sequence is monotonically decreasing and $x_n \searrow 0 $. On the other hand, at the $n$-th iteration, the resulting proposal density is given by
$$
f^*_n(y) = \sum_{i=0}^{n} F(x_{i+1}) h(y) \cdot \mathbf{1}_{\braces{ x_i \leq y < x_{i+1}  }}.
$$
Then the point such that $A_l = (A_l + A_r) \cdot \epsilon_0 / 2$ can be represented by a sequence $\braces{y_n}$ such that:
$$
\frac{ \int_{0}^{y_n} f^*_n(y) dy }
{ \int_{0}^{+\infty} f^*_n(y) dy } = \frac{\epsilon_0}{2}.
$$

Now we show that $\{y_n\}$ is a monotonically increasing sequence. Note that 
$$
f^*_{n+1}(y) = 
\begin{cases}
f^*_{n}(y), & y \geq x_{n+1}, \\
(1 - \epsilon_0 / 2) f^*_{n}(y), & 0 \leq y < x_{n+1}.
\end{cases}
$$
When $x_{n+1} < y_n \leq x_n$, we have
\begin{align*}
\begin{split}
 \frac{ \int_{0}^{y_{n}} f^*_{n+1}(y) dy }
 { \int_{0}^{+\infty} f^*_{n+1}(y) dy } 
&= \frac{ \int_{0}^{x_{n+1}} f^*_{n+1}(y) dy + \int_{x_{n+1}}^{y_{n}} f^*_{n+1}(y) dy  }
 { \int_{0}^{x_{n+1}} f^*_{n+1}(y) dy + \int_{x_{n+1}}^{y_{n}} f^*_{n+1}(y) dy + 
 \int_{y_{n}}^{+\infty} f^*_{n+1}(y) dy}  \\
&= \frac{ (1 - \epsilon_0/2) \int_{0}^{x_{n+1}} f^*_{n}(y) dy + \int_{x_{n+1}}^{y_{n}} f^*_{n}(y) dy  }
 { (1 - \epsilon_0/2) \int_{0}^{x_{n+1}} f^*_{n}(y) dy + \int_{x_{n+1}}^{y_{n}} f^*_{n}(y) dy + 
 \int_{y_{n}}^{+\infty} f^*_{n}(y) dy} \\
&< \frac{ \epsilon_0/2 \int_{0}^{x_{n+1}} f^*_{n}(y) dy +
    (1 - \epsilon_0/2) \int_{0}^{x_{n+1}} f^*_{n}(y) dy + \int_{x_{n+1}}^{y_{n}} f^*_{n}(y) dy  }
    { \epsilon_0/2 \int_{0}^{x_{n+1}} f^*_{n}(y) dy +  (1 - \epsilon_0/2) \int_{0}^{x_{n+1}} f^*_{n}(y) dy + \int_{x_{n+1}}^{y_{n}} f^*_{n}(y) dy + 
    \int_{y_{n}}^{+\infty} f^*_{n}(y) dy} \\
&= \frac{ \int_{0}^{y_{n}} f^*_{n}(y) dy }
 { \int_{0}^{+\infty} f^*_{n}(y) dy },
\end{split}
\end{align*}
implying that we must have $y_{n+1} > y_n$ such that
$$
\frac{ \int_{0}^{y_{n+1}} f^*_{n+1}(y) dy }
{ \int_{0}^{+\infty} f^*_{n+1}(y) dy } =
\frac{ \int_{0}^{y_n} f^*_n(y) dy }
{ \int_{0}^{+\infty} f^*_n(y) dy } = 
\frac{\epsilon_0}{2}.
$$
Similarly, when $y_n \leq x_{n+1} < x_n$, we have
\begin{align*}
\begin{split}
 \frac{ \int_{0}^{y_{n}} f^*_{n+1}(y) dy }
 { \int_{0}^{+\infty} f^*_{n+1}(y) dy } 
&= \frac{ \int_{0}^{y_{n}} f^*_{n+1}(y) dy  }
 { \int_{0}^{y_n} f^*_{n+1}(y) dy + \int_{y_n}^{x_{n+1}} f^*_{n+1}(y) dy + 
 \int_{x_{n+1}}^{+\infty} f^*_{n+1}(y) dy}  \\
&= \frac{ (1 - \epsilon_0/2) \int_{0}^{y_{n}} f^*_{n}(y) dy}
 { (1 - \epsilon_0/2) \int_{0}^{y_{n}} f^*_{n}(y) dy + (1 - \epsilon_0/2) \int_{y_{n}}^{x_{n+1}} f^*_{n}(y) dy + 
 \int_{x_{n+1}}^{+\infty} f^*_{n}(y) dy} \\
&< \frac{ \int_{0}^{y_{n}} f^*_{n}(y) dy   }
    { \int_{0}^{y_{n}} f^*_{n}(y) dy + \int_{y_{n}}^{x_{n+1}} f^*_{n}(y) dy + 
 \int_{x_{n+1}}^{+\infty} f^*_{n}(y) dy} \\
&= \frac{ \int_{0}^{y_{n}} f^*_{n}(y) dy }
 { \int_{0}^{+\infty} f^*_{n}(y) dy },
\end{split}
\end{align*}
implying that we must have $y_{n+1} > y_n$ such that
$$
\frac{ \int_{0}^{y_{n+1}} f^*_{n+1}(y) dy }
{ \int_{0}^{+\infty} f^*_{n+1}(y) dy } =
\frac{ \int_{0}^{y_n} f^*_n(y) dy }
{ \int_{0}^{+\infty} f^*_n(y) dy } = 
\frac{\epsilon_0}{2}.
$$

If $x_1 \leq y_1$, then we achieve $A_l \leq (A_l + A_r) \cdot \epsilon_0 / 2$ after the first iteration. Otherwise, if $x_1 > y_1$, then since $\braces{x_n}$ is monotonically decreasing to 0, and $\braces{y_n}$ is monotonically increasing, there must exist $n_0 > 1$ such that $x_{n_0} \leq y_{n_0}$ so at the $n_0$-th iteration we achieve $A_l \leq (A_l + A_r) \cdot \epsilon_0 / 2$.

\begin{table}[ht]
\centering
\caption{The mean acceptance rates (standard deviation) with $N = 30$ replicates for different $(\beta, \lambda=-0.001)$ using Algorithm \ref{algo:find-points}. }
\begin{tabular}{lllll}
\hline
Desired Acceptance Rate  & \multicolumn{1}{c}{0.25} & \multicolumn{1}{c}{0.50} & \multicolumn{1}{c}{0.75} & \multicolumn{1}{c}{0.90} \\ \hline
$\beta = 0.0001$ & 0.797 (0.004)            & 0.830 (0.003)           & 0.932 (0.002)            & 0.973 (0.002)           \\
$\beta = 0.001$  & 0.790 (0.004)            & 0.838 (0.004)           & 0.933 (0.003)            & 0.972 (0.001)           \\
$\beta = 0.01$   & 0.740 (0.004)            & 0.856 (0.003)           & 0.927 (0.002)            & 0.969 (0.002)           \\
$\beta = 0.1$    & 0.771 (0.003)            & 0.835 (0.003)           & 0.912 (0.003)            & 0.962 (0.002)           \\ \hline
\end{tabular}
\label{tab:acc-rate-2}
\end{table}

The table \ref{tab:acc-rate-2} shows the simulated acceptance rates for different $(\lambda, \beta)$ combinations given several desired $1 - \epsilon_0$ using the Algorithm \ref{algo:find-points}. Compared with the first column in Table \ref{tab:acc-rate-1} where we used the naive version, we see that the new algorithm induces a profound improvement in terms of higher acceptance rates. However, we see that the simulated acceptance rates are much higher than the user-specified rates especially when $1-\epsilon_0$ is smaller. This implies that our upper bound $\epsilon_0/2 + (1 - \epsilon_0/2) \cdot \epsilon_0 / 2$ for the actual overall rejection rate is too conservative and too loose, and resources are wasted to set up more cutoff points than necessary. One possible ad-hoc strategy in practice is to feed in $2\epsilon_0$ instead of $\epsilon_0$ to Algorithm \ref{algo:find-points}, and the simulation results are shown in Table \ref{tab:acc-rate-3}. We see that the simulated acceptance rate is still larger than desired but not that large here.

\begin{table}[ht]
\centering
\caption{The mean acceptance rates (standard deviation) with $N = 30$ replicates for different $(\beta, \lambda=-0.001)$ by feeding in Algorithm \ref{algo:find-points} with $2\epsilon_0$. }
\begin{tabular}{lllll}
\hline
Desired Acceptance Rate  & \multicolumn{1}{c}{0.25} & \multicolumn{1}{c}{0.50} & \multicolumn{1}{c}{0.75} & \multicolumn{1}{c}{0.90} \\ \hline
$\beta = 0.0001$ & 0.303 (0.002)            & 0.600 (0.005)           & 0.830 (0.005)            & 0.948 (0.002)           \\
$\beta = 0.001$  & 0.784 (0.004)            & 0.853 (0.004)           & 0.839 (0.003)            & 0.945 (0.003)           \\
$\beta = 0.01$   & 0.602 (0.004)            & 0.636 (0.004)           & 0.855 (0.003)            & 0.942 (0.002)           \\
$\beta = 0.1$    & 0.364 (0.003)            & 0.688 (0.005)           & 0.835 (0.004)            & 0.930 (0.003)           \\ \hline
\end{tabular}
\label{tab:acc-rate-3}
\end{table}

\begin{figure}[ht]
    \centering
    \includegraphics[width=0.7\textwidth]{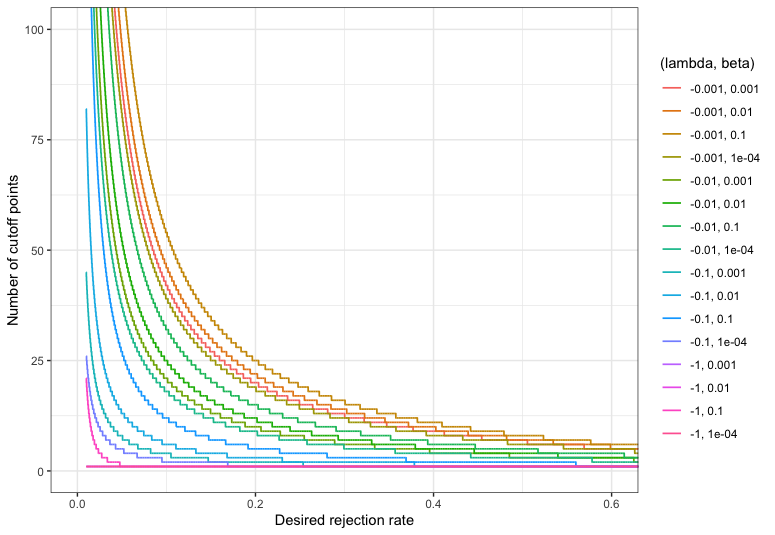}
    \caption{The histogram of varieties from $\mathcal{GIG}(-0.1, 1, 1)$ generated using Algorithm \ref{algo:find-points}. The blue curve is the actual density. }
    \label{fig:monotone}
\end{figure}

In addition to the control of acceptance rates, the Algorithm \ref{algo:find-points} (and its ad-hoc version) also provides a tool to see the trade-off between the rejection rate and the number of cutoff points. In Figure \ref{fig:monotone}, we divide the rejection rate range $(0,1]$ into 50000 equispaced intervals and compute the corresponding number of cutoff points for 16 different $\mathcal{GIG}(\lambda, \beta)$. Based on the plot, we see that the number of cutoff points is a monotonically decreasing function of the desired rejection rate $\epsilon_0$. Particularly, when $\epsilon_0 \rightarrow 1$, the number of cutoff points goes to 0; whereas when $\epsilon_0 \rightarrow 0$, then number of cutoff points goes to infinity.

It is not surprising that a lower desired rejection rate would lead to more cutoff points, as intuitively it may take more iterations for each piece of the proposal density to be tighter to $f(y)$. On the other hand, a smaller number of cutoff points may lead to a higher overall rejection rate, and in the extreme case when we have no cutoff points the Algorithm \ref{algo:find-points} degenerates to the native version.

\subsection{Cutoff points given the number of cutoff points}

We have investigated the strategy to choose cutoff points when the rejection rate is given; however, on the other end of the trade-off spectrum, it is also critical to develop a strategy that adaptively chooses cutoff points given the desired number of points. We hope that the specific cutoff points $\mathbf k$ we obtain given the number of cutoff points $K$, can be consistent with those cutoff points $\mathbf k'$ given the desired rejection rate $\epsilon_0$. In other words, the points $\mathbf k$ should achieve similar guaranteed acceptance rates as those $\mathbf k'$ such that $|\mathbf k| = |\mathbf k'|$. Similarly, the points $\mathbf k'$ should have approximately the same length as that of $\mathbf k$ if $\mathbf k$ and $\mathbf k'$ enjoy similar rejection rate bounds.

Fortunately, the number of cutoff points $K$ is a monotonically decreasing function of the desired rejection rate $\epsilon_0$ as shown in the last section. We call this function $K = K^*(\epsilon_0)$. With this trade-off we may take advantage of binary search to find
$$
\epsilon_K = \inf\braces{ \epsilon: K = K^*(\epsilon) },
$$
and the corresponding exact cutoff points (see Algorithm \ref{algo:find-points-num}).

\begin{algorithm}[!htb] 
    \caption{Find cutoff points given the number of cutoff points \\
    \texttt{find.cutoff.under.fix.num()} } 
    \label{algo:find-points-num}
    \KwIn{Parameters $\lambda$, $\beta$, the number of cutoff points $K$, and threshold $t_0 = 1\times 10^{-6}$;}
    Initialization: $l \leftarrow 0$, $r \leftarrow 1$;
    \\
    \While{ $r - l > t_0$ }{
        $m \leftarrow (r + l) / 2$; \\
        $\mathbf k \leftarrow \texttt{find.cutoff.under.rej.rate(}\lambda, \beta, m \texttt{)} $; \\
        $K_m \leftarrow |\mathbf k|$; \\
        \uIf{$K_m < K$}{
        $r \leftarrow m$;
        }
        \Else{
         $l \leftarrow m$; \\
        }
    }
    CUTPOINTS $\leftarrow \texttt{find.cutoff.under.rej.rate(}\lambda, \beta, l \texttt{)} $; \\
        \KwOut{A vector of cutoff points CUTPOINTS.  }
\end{algorithm}

The Table \ref{tab:acc-rate-4} shows the simulated acceptance rates when the numbers of cutoff points are pre-specified. Note that for any cell in this table the actual number of cutoff points coincides with the corresponding desired number. Consequently, it is clear that with more cutoff points the actual rejection rate decreases, which again reveals their trade-off relationship.

\begin{table}[ht]
\centering
\caption{The mean acceptance rates (standard deviation) with $N = 30$ replicates for different $(\beta, \lambda=-0.001)$ by specifying different desired numbers of cutoff points. }
\begin{tabular}{lcccc}
\hline
Desired \# (Cutoff) & 1             & 5             & 10            & 50            \\ \hline
$\beta = 0.0001$    & 0.018 (0.000) & 0.723 (0.004) & 0.765 (0.005) & 0.959 (0.002) \\
$\beta = 0.001$     & 0.014 (0.000) & 0.713 (0.004) & 0.756 (0.004) & 0.955 (0.002) \\
$\beta = 0.01$      & 0.009 (0.000) & 0.689 (0.004) & 0.740 (0.004) & 0.948 (0.002) \\
$\beta = 0.1$       & 0.005 (0.000) & 0.629 (0.004) & 0.711 (0.004) & 0.929 (0.002) \\ \hline
\end{tabular}
\label{tab:acc-rate-4}
\end{table}

Note that the Algorithm \ref{algo:find-points-num} should be much less efficient than the Algorithm \ref{algo:find-points}, since the function \texttt{find.cutoff.under.rej.rate()}, in general, need to be called for multiple times during the procedure of binary search in Algorithm \ref{algo:find-points-num} when we make use of the non-increasing function $K = K^*(\epsilon_0)$. Sometimes people may need to compare the time complexity of algorithms under different settings by specifying a benchmark, such as a common desired rejection rate (an upper bound for the rejection rate) or a common number of cutoff points. Since it is efficient to evaluate the function $K = K^*(\epsilon_0)$ given $\epsilon_0$ while rather inefficient to find the pre-image of $K^*(\cdot)$ given $K$, we should always use a common desired rejection rate $\epsilon_0$. 

Why do we design the Algorithm \ref{algo:find-points-num}, then, if the binary search procedure is so time-costly?
From Figure \ref{fig:monotone}, we see that when the desired rejection rate $\epsilon_0$ is close to 0, a tiny shift downward in $\epsilon_0$ would result in an extremely large increment in the number of cutoff points. In this case, the time spent on setting up cutoff points might exceed the time saved from a small rejection rate so we have to be careful in terms of controlling the most time-consuming part. In summary, under the setting when the number of cutoff points is significantly effective on the running time and it is affordable to run the binary search, we may use the function \texttt{find.cutoff.under.fix.num()} before the actual sampling procedure.

Additionally, there are cases when we need to compare our algorithm with other existing procedures where the actual rejection rate (or rejection constant) is chosen as a comparing metric. It would make no sense if we control the comparing metric in advance, and thus choosing a common number of cutoff points for each parameter is much more fair for the comparisons if the time complexity is not a major concern.

\section{The Sampling Procedure for GIG}
In this section, we are going to summarize and encapsulate the previous results into a combined sampling procedure for any $\mathcal{GIG}(\lambda, \psi, \chi)$ distributions when $\lambda\neq 0$ and $(\psi, \chi)$ satisfies the condition \ref{for:range}. As stated before we will consider two scenarios, one is the sampling procedure given the desired rejection rate, while the other is that given the number of cutoff points (see Algorithm \ref{algo:combined}).

\begin{algorithm}[!htb] 
    \caption{GIG variates generator   } 
    \label{algo:combined}
    \KwIn{Number of variates $N$, parameters $\lambda$, $\psi$, $\chi$. In addition, either the desired rejection rate $\epsilon_0$ or the desired number of cutoff points $K$ is needed;}
    \If{$\lambda > 0$}{
        $(\lambda, \psi, \chi) \leftarrow (-\lambda, \chi, \psi)$;  \\
        SignFlag $\leftarrow$ TRUE; \\
        }
    Reparameterization: $\alpha \leftarrow \sqrt{\psi / \chi}$, and $\beta \leftarrow \sqrt{\psi \chi}$;  \\
    The number of cutoff points:
    \uIf{$\epsilon_0$ is provided}{
        $\mathbf k \leftarrow \texttt{find.cutoff.under.rej.rate(}\lambda, \beta, \epsilon_0$\texttt{)};
    }\Else{
        $\mathbf k \leftarrow \texttt{find.cutoff.under.fix.num(}\lambda, \beta, K, t_0 = 1\times10^{-6}$\texttt{)}; \\
    }
    Extend the cutoff points: $\mathbf k \leftarrow [0, \mathbf k, \infty]$  \\
    Sample from $p_Y(y)$: $ \braces{Y_1, \dots, Y_N} \leftarrow$ \texttt{rejection.sampling($N, \lambda, \beta, \mathbf k$)};  \\
    Sample from $X|Y$: $X_j \leftarrow$ \texttt{truncated.gamma($-\lambda, \beta / 2, 1/Y_j$)} for $j = 1,\dots, N$;  \\
    GIG variate: $X_j \leftarrow 1/(\alpha X_j)$ for $j = 1,\dots, N$;  \\
    \If{SignFlag is TRUE}{
        $X_j \leftarrow 1/X_j$ for $j = 1,\dots, N$;  \\
    }
    \KwOut{The variates $\braces{X_1, \cdots, X_N} \overset{\rm i.i.d}{\sim} \mathcal{GIG}(\lambda, \psi, \chi)$.  }
\end{algorithm}

\subsection{Time Complexity Analysis}
In this section, we are going to delve into the time complexity of the combined algorithm to explore the relationship between running time and the user-specified tuple $(\epsilon_0, K)$.

The Algorithm \ref{algo:combined} can be decomposed in the following way:
\begin{enumerate}[label = \underline{Part \arabic*}:]
    \item The line 1 to line 5 and the line 11 take around constant time.
    
    \item The line 6 to line 10 are associated with the number of cutoff points. In particular, if we find the cutoff points given the desired rejection rate (through the function \texttt{find.cutoff.under.rej.rate()} in Algorithm \ref{algo:find-points}), then the time spent in this part can be expressed as $T_1 \cdot K^*(\epsilon_0)$ where $T_1$ is the average time of looking for each cutoff point, and $K^*(\epsilon_0)$ is the resulting number of cutoff points. Note that the actual value of $T_1$ depends on specific machine settings but the time used for each cutoff point should remain similar.
    
    \item The line 12 is associated with the rejection sampling procedure. Particularly, inside the Algorithm \ref{algo:rej-samp} (denoted by \texttt{rejection.sampling()}) the line 1 is of constant time. Additionally, the time spent on line 2 and line 3 can be expressed as $T_2 \cdot K^*(\epsilon_0)$ where $T_2$ is the average time on setting up each point. The time spent on line 4 to line 13 can be expressed as $T_3 \cdot N'$ where $T_3$ is the average time of drawing one variate from the proposal density, and $N'$ is the total number of draws (among which some are accepted and some are rejected). Similarly as above, the actual values of $T_2$ and $T_3$ depend on specific machine settings but the time used for each draw should remain similar.
    
    \item The line 13 to 17 are of linear time of $N$. In other words, the time spent on this part can be expressed as $T_4 \cdot N$, which solely depends on the number of variates. Still, the acutal value of $T_4$ dependes on specific machine settings but the time used for post-processing each variate should remain similar.
\end{enumerate}

Now we ignore the constant time consumption and consider the average generating time for each variate:
$$
\frac{T_{\rm total}}{N} =
\frac{ T_1 \cdot K^*(\epsilon_0) + T_2 \cdot K^*(\epsilon_0) + T_3 \cdot N' + T_4 \cdot N }{N} \approx
(T_1 + T_2) \cdot \frac{K^*(\epsilon_0)}{N} + T_3 \cdot \frac{1}{1-\epsilon_0} + T_4
$$
where we use the desired rejection rate to approximate the actual rejection rate $1 - N/N'$. Note that $K^*(\epsilon_0)$ is a decreasing function of $\epsilon_0$ so the first term is a decreasing function of $\epsilon_0$ as well; however, the second term is an increasing function of $\epsilon_0$. When $\epsilon_0$ goes close to 0 or 1 the average time tends to infinity, implying that there should be a minimum point in between such that $T_{\rm total} / N$ attains its minimum when $N$ is fixed.

We may also see the effect of varying $N$ when $\epsilon_0$ is fixed. When $N$ is small, for example when we are trying to draw only one sample for each iteration of a Gibbs sampler, the part $(T_1 + T_2) \cdot K^*(\epsilon_0)$ representing the cutoff-point-processing time might dominate the average generating time particularly when $\epsilon_0$ is small. However, when $N$ becomes large, for example when we are trying to obtain Monte Carlo estimates through drawing a large sample from GIG, the part $T_3 / (1 - \epsilon_0) + T_4$ will ultimately dominate the average generating time. 

In summary, given the same desired rejection rate $\epsilon_0$, a large sample size should be able to save the average time spent on pre-processing cutoff points as the pre-processing procedure need not be repeated for each subsequent iteration of rejection sampling. 

\begin{figure}[ht]
     \centering
     \begin{subfigure}[b]{0.49\textwidth}
         \centering
         \includegraphics[width=\linewidth]{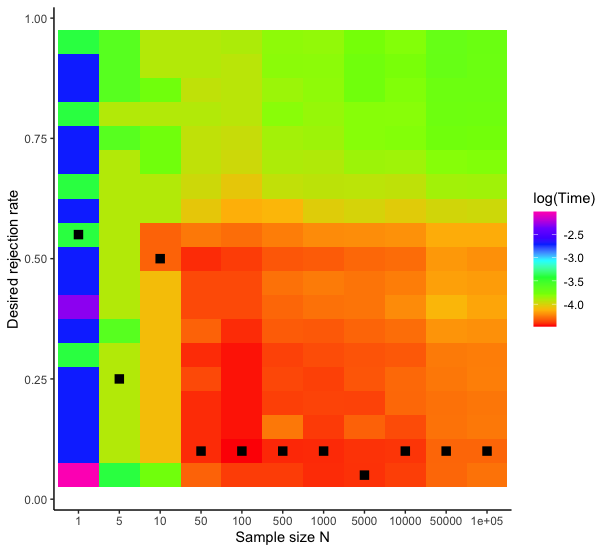}
         \caption{When $\lambda = 0.1$ and $\beta = 0.1$. }
         \label{sfig:time-1}
     \end{subfigure}
     \hfill
     \begin{subfigure}[b]{0.49\textwidth}
         \centering
         \includegraphics[width=\linewidth]{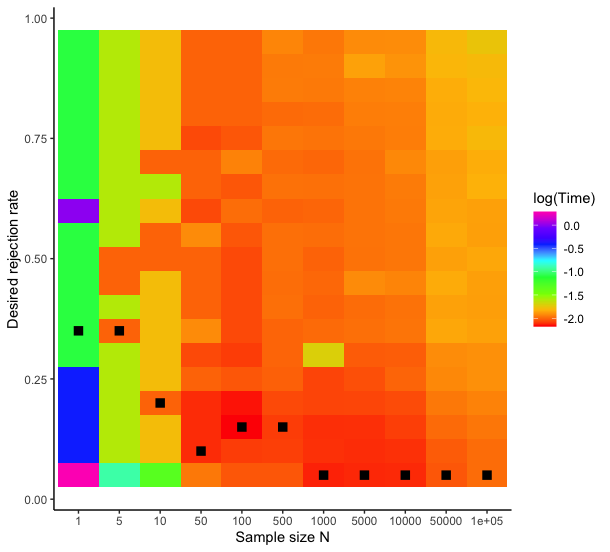}
         \caption{When $\lambda = 1$ and $\beta = 1$. }
         \label{sfig:time-2}
     \end{subfigure}
     \caption{The logarithm of average generating time (sec) under different desired rejection rates $\epsilon_0$ and sample sizes $N$. The black squares represent the minimal values in each column. }
     \label{fig:time-comp}
\end{figure}

The Figure \ref{fig:time-comp} shows simulated average generating time $T_{\rm total} / N$ under different desired rejection rates $\epsilon_0$ and different samples sizes $N$, where the black squares represent the minimal values in each column. The value in each cell is computed as the average value over 30 replicates. Generally speaking, for a small sample size ($N \leq 10$), the best desired rejection rate $\epsilon_0$ which induces the shortest average running time would be around 0.5; while for a large sample size ($N > 1000$), the corresponding best desired rejection rate $\epsilon_0$ would be around 0.05 to 0.1. In addition, the average running time given $\epsilon_0$ is roughly decreasing with increasing $N$. Particularly, a larger sample size $N$ and a smaller $\epsilon_0$ would lead to an overall shorter running time: the smallest values cluster in the right bottom corner of the plots.

These plots again show the trade-off between rejection rate and number of cutoff points. When we need only one variate, it is not worthwhile to spend too much time on setting up cutoff points even though it may result in higher rejection rate. On the other hand, when we need a large sample, the benefits of low rejection rate would exceed the average loss on the pre-processing procedure that deals with cutoff points so a low rejection rate is the best option.

\subsection{Our Sampler v.s. Ratio-of-uniforms}

Previously we provided time complexity analysis and performed running time comparisons between different scenarios for our own approach. Nevertheless, it is always a concern that how our approach would outperform other existing algorithms. As we mentioned in Section 1, \cite{hormann2014generating} improved a commonly used ratio-of-uniforms method for generating GIG variates and provide their uniformly bounded rejection constants for the particularly interested range $\lambda \in (0, 1.5]$ and $\beta \in (0, 1.5]$. The rejection constant can be estimated by \#(variates draw from proposal density) / \#(variates being accepted).

In order to align with their results, we have to treat the rejection constant as the metric as well. Note that in our approach, the rejection constants can be upper bounded by $1 / (1-\epsilon_0)$ and hence are automatically controlled if we specify the desired rejection rate $\epsilon_0$, although might at the price of extremely large number of cutoff points. Another strategy of comparison, with the goal of a fair game, is to fix several different numbers of cutoff points and then compare the rejection constants.

The Figure \ref{fig:out-comp} shows the estimated rejection constants when we sample 50000 independent variates for each parameter $(\lambda, \beta)$. One immediate observation is that the sampling procedure is particularly harder when $\lambda$ is close to 0, where the rejection constants become abruptly large. In addition, higher $\beta$ leads to larger rejection constant while higher $\lambda$ leads to lower rejection constant. Finally, as the number of cutoff points increase, the rejection constants are decreasing. Specifically, the rejection constants are uniformly bounded and smaller than the results in~\cite{hormann2014generating} when $K \geq 20$.

\begin{figure}[ht]
     \centering
     \begin{subfigure}[b]{0.46\textwidth}
         \centering
         \includegraphics[width=\linewidth]{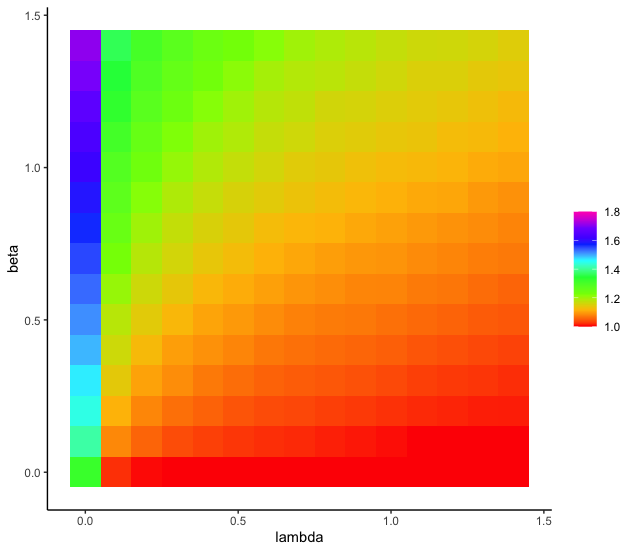}
         \caption{$K = 10$. }
         \label{sfig:comp-1}
     \end{subfigure}
     \hfill
     \begin{subfigure}[b]{0.46\textwidth}
         \centering
         \includegraphics[width=\linewidth]{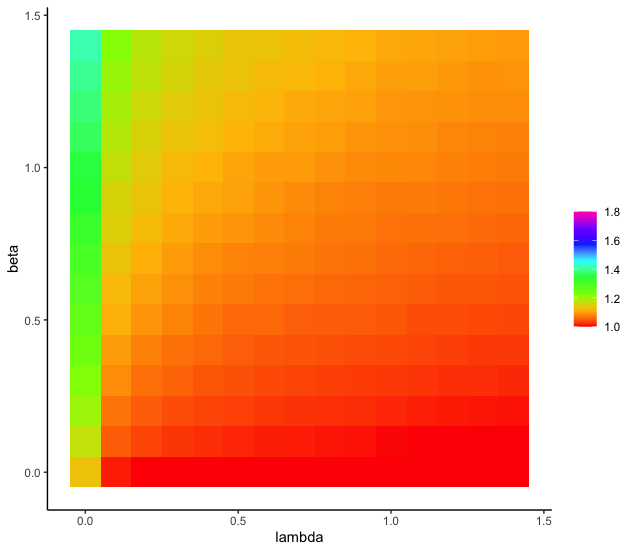}
         \caption{$K = 20$. }
         \label{sfig:comp-2}
     \end{subfigure}
     \hfill
     \begin{subfigure}[b]{0.46\textwidth}
         \centering
         \includegraphics[width=\linewidth]{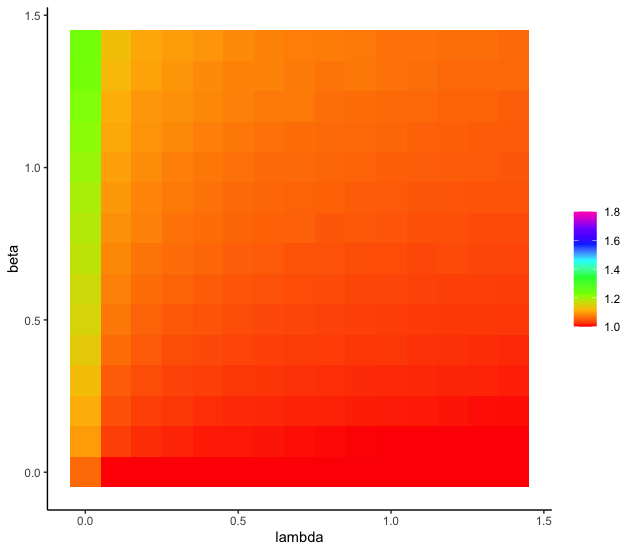}
         \caption{$K = 40$. }
         \label{sfig:comp-3}
     \end{subfigure}
     \hfill
     \begin{subfigure}[b]{0.46\textwidth}
         \centering
         \includegraphics[width=\linewidth]{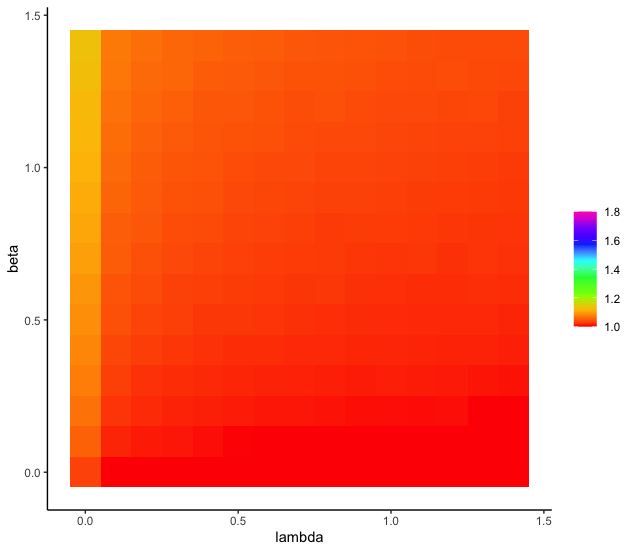}
         \caption{$K = 80$. }
         \label{sfig:comp-4}
     \end{subfigure}
     \hfill
     \begin{subfigure}[b]{0.46\textwidth}
         \centering
         \includegraphics[width=\linewidth]{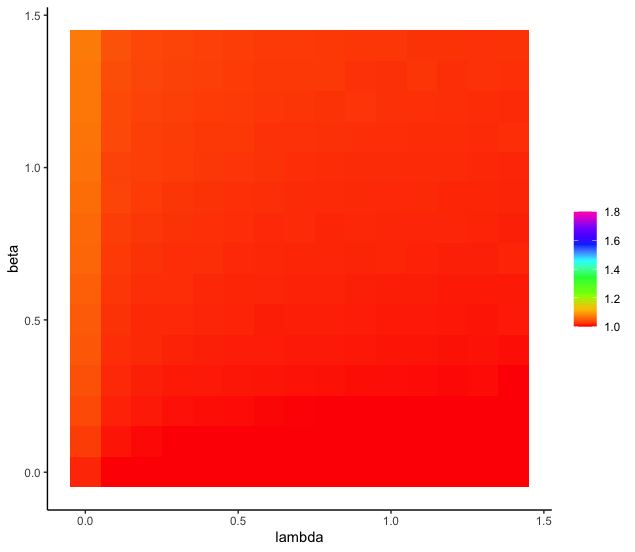}
         \caption{$K = 160$.}
         \label{sfig:comp-5}
     \end{subfigure}
     \hfill
     \begin{subfigure}[b]{0.46\textwidth}
         \centering
         \includegraphics[width=\linewidth]{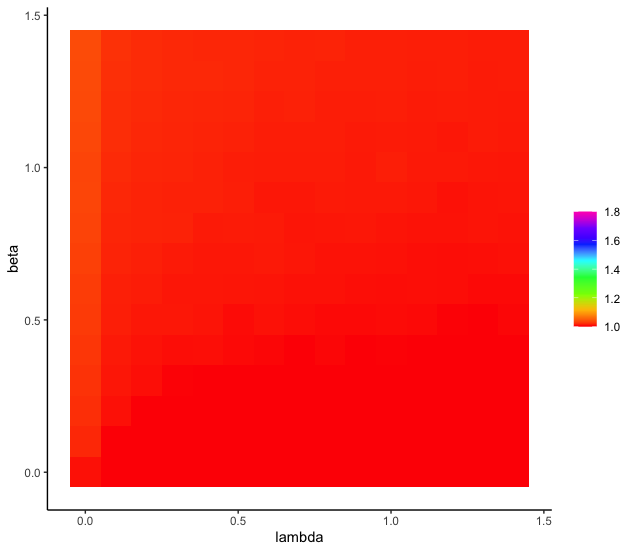}
         \caption{$K = 320$. }
         \label{sfig:comp-6}
     \end{subfigure}
     \caption{The rejection constants under different parameters $(\lambda, \beta)$ and different numbers of cutoff points. }
     \label{fig:out-comp}
\end{figure}

\section{Discussion}
We have proposed a new generation method for GIG distributions that adaptively controls the rejection rate given a user-specified upper bound. Rejection sampling algorithms in general have to balance the acceptance rate with setup efficiency. In our method, this trade-off is demonstrated by the complexity of the piecewise density $f^*(y)$, which can be manually controlled by users through the desired rejection rate $\epsilon_0$ or the desired number of cutoff points $K$. Our proposed sampling approach can deal with both the parameter varying case and the large sample case efficiently, through the powerful automatic adjustment of proposal densities. Simulation results show that when sampling only one variate at a time, the desired rejection rate $\epsilon_0$ can be set around 0.25 to 0.5, while when sampling more than 1000 samples at a time the $\epsilon_0$ can be set around 0.05 to 0.1.

Our algorithm is competitive among other existing packages that generate GIG distributions. Our simulation experiments show that with more than 20 cutoff points, we can achieve uniformly smaller rejection constants than the most recent and most efficient algorithm \citep{hormann2014generating}. However, we never managed to outperform the function \texttt{GIGrvg::rgig()} when we are implementing our algorithm in Rcpp\footnote{See \texttt{https://github.com/Xiaozhu-Zhang1998/GIGcpp}.} and comparing their actual running time in the same R environment, although for certain hard cases (e.g. $\lambda = 10^{-5}$, $\chi = 10^{-7}$, $\psi = 1$) we spent less time than the function $\texttt{ghyp::rgig()}$. We do realize that much improvement is still needed to optimize our current Rcpp codes.

In addition to the specific quasi-density $f(y)$, we believe that our piecewise rejection sampling procedure can be extended to any other target density as long as it can be written as the product of two functions, at least one of which is monotonic increasing and both of which have fast inversion formulae (or algorithms).





\newpage

\appendix
\section*{Appendix}
\label{app:theorem}



\noindent \textbf{Proof of Lemma \ref{lem:trunc}}
\begin{proof}
\ 
\begin{enumerate}[label = (\arabic*)]
    \item For any $y\in\mathbb{R}$,
    \begin{align*}
    \begin{split}
      \mathbb{P}(Y-p \leq y) 
    &= \mathbb{P}(Y \leq y + p)  \\
    &= \left( 1 - e^{-(y+p)} \right) \cdot \mathbf{1}_{\{y + p > 0\}}  \\
    &= \frac{e^p - e^{-y}}{e^p} \cdot \mathbf{1}_{\{y > -p\}},
    \end{split}
    \end{align*}
    which is exactly the CDF of $\mathsf{Ex}(1) \cdot \mathbf{1}_{\{x>-p\}}$.

    \item For any $u\in(1-e^p, 1)$,
    \begin{align*}
    \begin{split}
      \mathbb{P}(U \leq u) 
    &= \mathbb{P}(-\log(1-U) \leq -\log(1-u) )  \\
    &= \frac{e^p -e^{\log(1-u)} }{e^p} \cdot \mathbf{1}_{\{-\log(1-u) >-p\}}  \\
    &= \frac{e^p - 1 + u}{e^p} \cdot \mathbf{1}_{\{1 - e^p < u < 1\}},
    \end{split}
    \end{align*}
    which is exactly the CDF of $\mathsf{Unif}(1-e^p, 1)$.
    
\end{enumerate}
\end{proof}

\noindent \textbf{Proof of Lemma \ref{lem:rej-samp}}
\begin{proof}
\ 
\begin{enumerate}[label = (\arabic*)]
    \item For any $a < y <b$, we have
    \begin{align*}
    \begin{split}
        \mathbb{P}(Y + a \leq y)
    &= \mathbb{P}(Y \leq y - a)  \\
    &= \frac{1 - e^{-\lambda(y-a)}} {1 - e^{-\lambda(b-a)}}  \\
    &= \frac{e^{-\lambda a} - e^{-\lambda y }} {e^{-\lambda a} - e^{-\lambda b}},
    \end{split}
    \end{align*}
    which is exactly of the CDF of $\mathsf{Ex}(\lambda) \cdot \mathbf{1}_{\{ a<y<b \}}$.
    
    \item Denote $\lfloor y \rfloor_a = \max\{n\in\mathbb{N}: an \leq y\} $, then $Y\mod a = Y - a \lfloor Y \rfloor_a$. For any $ t > 0$, we have
    \begin{align*}
    \begin{split}
        \mathbb{P}(Y\mod a \leq t)
    &= \mathbb{P}(Y - a \lfloor Y \rfloor_a \leq t ) \\
    &= \sum_{n=0}^\infty \mathbb{P}(Y - an \leq t, \lfloor Y \rfloor_a = n)  \\
    &= \sum_{n=0}^\infty \mathbb{P}(an < Y \leq an + t ) \\
    &= \sum_{n=0}^\infty \left[ e^{-\lambda an} - e^{-\lambda (an + t)}  \right] \\
    &= \frac{1 - e^{-\lambda t}}
    {1 - e^{\lambda a}},
    \end{split}
    \end{align*}
    which is the exactly the CDF of $\mathsf{Ex}(\lambda) \cdot \mathbf{1}_{0 < y< a}$.
\end{enumerate}
\end{proof}

\newpage
\vskip 0.2in
\bibliography{ref}

\end{document}